\documentclass{llncs}
\usepackage{amsmath,amssymb,amsfonts}
\usepackage{booktabs,tabularx}
\usepackage{url,xspace,soul,wrapfig}
\usepackage{graphicx,hyperref}
\usepackage{subcaption}
\usepackage{paralist,enumerate}
\usepackage[textsize=footnotesize]{todonotes}
\usepackage{enumerate}
\usepackage{soul,lineno}

\let\doendproof\endproof
\renewcommand\endproof{~\hfill$\qed$\doendproof}

\makeatother

\newcommand{\Ss}{\mathcal{S}}
\newcommand{\bo}{\sigma}
\newcommand{\co}{\chi}
\newcommand{\auG}{\widehat{G}}

\newcommand{\auC}{\widehat{C}}
\newcommand{\auP}{\widehat{P}}
\newcommand{\auGamma}{\widehat{\Gamma}}

\newcommand{\bcw}[1]{\overset{\curvearrowright}{\rho_{#1}}}
\newcommand{\bccw}[1]{\overset{\curvearrowleft}{\rho_{#1}}}
\newcommand{\rcw}[1]{\overset{\curvearrowright}{\delta_{#1}}}
\newcommand{\rccw}[1]{\overset{\curvearrowleft}{\delta_{#1}}}
\newcommand{\nr}[2]{d_r(#1,#2)}
\newcommand{\nd}[2]{d_d(#1,#2)}

\newcommand{\myref}[1]{{\textbf{#1}}}

\title{Universal Slope Sets for\\ Upward Planar Drawings}

\author{Michael A. Bekos\inst{1}, 
Emilio Di~Giacomo\inst{2}, 
Walter Didimo\inst{2},\\ 
Giuseppe Liotta\inst{2}, 
Fabrizio Montecchiani\inst{2}}

\institute{
Institut f{\"u}r Informatik, Universit{\"a}t T{\"u}bingen, Germany
\\
\email{bekos@informatik.uni-tuebingen.de}
\and
Dipartimento di Ingegneria, Universit\`a degli Studi di Perugia, Italy
\\
\email{name.surname@unipg.it}
}

\begin{document}
\maketitle

\pagenumbering{arabic}

\begin{abstract} 
We prove that every set $\Ss$ of $\Delta$ slopes containing the horizontal slope is \emph{universal} for $1$-bend upward planar drawings of bitonic $st$-graphs with maximum vertex degree $\Delta$, i.e., every such digraph admits a $1$-bend upward planar drawing whose edge segments use only slopes in $\Ss$. This result is worst-case optimal in terms of the number of slopes, and, for a suitable choice of $\Ss$, it gives rise to drawings with worst-case optimal angular resolution. In addition, we prove that every such set $\Ss$ can be used to construct $2$-bend upward planar drawings of $n$-vertex planar $st$-graphs with at most $4n-9$ bends in total. Our main tool is a constructive technique that~runs~in~linear~time.  
\end{abstract}

\section{Introduction}
Let $G$ be a graph with maximum vertex degree $\Delta$. The \emph{$k$-bend planar slope number} of $G$ is the minimum number of slopes for the edge segments needed to construct a $k$-bend planar drawing of $G$, i.e., a planar drawing where each edge is a polyline with at most $k \ge 0$ bends. Since no more than two edge segments incident to the same vertex can use the same slope, $\lceil \Delta/2 \rceil$ is a trivial lower bound for the $k$-bend planar slope number of $G$,  irrespectively of $k$. Besides its theoretical interest, this problem forms a natural extension of two well-established graph drawing models: The \emph{orthogonal}~\cite{DBLP:journals/comgeo/BiedlK98,orthoChapterHandbook,DBLP:journals/siamcomp/GargT01,DBLP:journals/siamcomp/Tamassia87} and the \emph{octilinear} drawing models~\cite{DBLP:journals/jgaa/BekosG0015,DBLP:conf/latin/Bekos0016,DBLP:journals/jgaa/BodlaenderT04,Noellenburg05}, which both have several applications, such as in VLSI and floor-planning~\cite{DBLP:conf/focs/Leiserson80,DBLP:journals/tc/Valiant81}, and in metro-maps and map-schematization~\cite{DBLP:journals/vlc/HongMN06,DBLP:journals/tvcg/NollenburgW11,DBLP:journals/tvcg/StottRMW11}. Orthogonal drawings use only $2$ slopes for the edge segments ($0$ and $\frac{\pi}{2}$), while octilinear drawings use no more than $4$ slopes ($0$, $\frac{\pi}{4}$, $\frac{\pi}{2}$, and $\frac{3\pi}{4}$); consequently, they are limited to graphs with $\Delta \le 4$~and~$\Delta \le 8$,~respectively.

These two drawing models have been generalized to graphs with arbitrary maximum vertex degree $\Delta$ by Keszegh et al.~\cite{DBLP:journals/siamdm/KeszeghPP13}, who proved that every planar graph admits a $2$-bend planar drawing with $\lceil \Delta/2 \rceil$ equispaced slopes. As a witness of the tight connection between the two problems, the result by Keszegh et al. was built upon an older result for orthogonal drawings of degree-$4$ planar graphs by Biedl and Kant~\cite{DBLP:journals/comgeo/BiedlK98}. In the same paper, Keszegh et al. also studied the $1$-bend planar slope number and showed an upper bound of $2 \Delta$ and a lower bound of $\frac{3}{4}(\Delta - 1)$ for this parameter. The upper bound has been recently improved, initially by Knauer and Walczak~\cite{DBLP:conf/latin/KnauerW16} to $\frac{3}{2}(\Delta - 1)$ and subsequently by Angelini et al.~\cite{DBLP:conf/compgeom/AngeliniBLM17} to $\Delta-1$. Angelini et al. actually proved a stronger result: Given \emph{any} set $\Ss$ of $\Delta-1$ slopes, every planar graph with maximum vertex degree $\Delta$ admits a 1-bend planar drawing whose edge segments use only slopes in $\Ss$. Any such slope set is hence called \emph{universal} for $1$-bend planar drawings. This result simultaneously establishes the best-known upper bound on the $1$-bend planar slope number of planar graphs and the best-known lower bound on the angular resolution of $1$-bend planar drawings, i.e., on the minimum angle between any two edge segments incident to the same vertex. Indeed, if the slopes in $\Ss$ are equispaced, the resulting drawings have angular resolution at least $\frac{\pi}{\Delta-1}$.

In this paper we study slope sets that are universal for $k$-bend \emph{upward} planar drawings of directed graphs (or digraphs for short). Recall that in an upward drawing of a digraph $G$, every edge $(u,v)$ is drawn as a $y$-monotone non-decreasing curve from $u$ to $v$. Also, $G$ admits an upward planar drawing if and only if it is a subgraph of a planar $st$-graph~\cite{DBLP:journals/tcs/BattistaT88,DBLP:journals/dm/Kelly87}. As such drawings are common for representing planar digraphs, they have been extensively studied in the literature  (see, e.g.,~\cite{DBLP:journals/siamcomp/BertolazziBMT98,DBLP:journals/jea/ChimaniZ15,DBLP:reference/algo/Didimo16,DBLP:journals/siamcomp/GargT01,chapterHandbook}). A preliminary result for this setting is due to Di Giacomo et al.~\cite{DBLP:conf/gd/GiacomoLM16}, who proved that every series-parallel digraph with maximum vertex degree $\Delta$ admits a  $1$-bend upward planar drawing that uses at most $\Delta$ slopes, and this bound on the number of slopes is worst-case optimal. Notably, their construction gives rise to drawings with optimal angular resolution $\frac{\pi}{\Delta}$ (but it uses a predefined set of slopes). Upward drawings with one bend per edge and few slopes have also been studied for posets by Czyzowicz et al.~\cite{Czyzowicz1990}.

\begin{figure}[t]
\centering
\begin{subfigure}{0.45\textwidth}
	\centering
	\includegraphics[width=\textwidth,page=1]{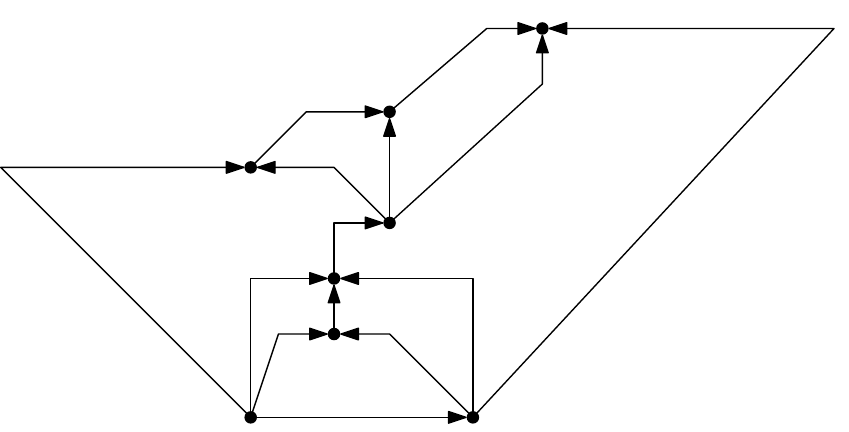}
	\caption{}
	\label{fi:example-1}
\end{subfigure}
\hfil
\begin{subfigure}{0.45\textwidth}
	\centering
	\includegraphics[width=\textwidth,page=2]{figs/intro}
	\caption{}
	\label{fi:example-2}
\end{subfigure}
\caption{(a) A $1$-bend upward planar drawing of a bitonic $st$-graph, and (b) a $2$-bend upward planar drawing of a planar $st$-graph, both defined on a slope set $\Ss=\{-\frac{\pi}{4},0,\frac{\pi}{4},\frac{\pi}{2},\pi\}$.}
\end{figure}

\smallskip\noindent{\bf Contribution.} We extend the study of universal sets of slopes to upward planar drawings, and present the first constructive technique that works for all planar $st$-graphs. This technique exploits a linear ordering of the vertices of a planar digraph introduced by Gronemann~\cite{DBLP:conf/gd/Gronemann16}, called \emph{bitonic $st$-ordering} (see also Section~\ref{sec:preliminaries}). We show that any set $\Ss$ of $\Delta$ slopes containing the horizontal slope is universal for $1$-bend upward planar drawings of  degree-$\Delta$ planar digraphs having a bitonic $st$-ordering (Section~\ref{sec:1bend}). We remark that the size of $\Ss$ is worst-case optimal~\cite{DBLP:conf/gd/GiacomoLM16} and, if the slopes of $\Ss$ are chosen to be equispaced, the angular resolution of the resulting drawing is at least $\frac{\pi}{\Delta}$ (also optimal); see Fig.~\ref{fi:example-1} for an illustration. We then extend our construction to all planar $st$-graphs by using two bends on a restricted number of edges (Section~\ref{sec:2bend}). More precisely, we show that, given a set $\Ss$ of $\Delta$ slopes containing the horizontal slope, every $n$-vertex upward planar digraph with maximum vertex degree $\Delta$ has a $2$-bend upward planar drawing that uses only slopes in $\Ss$ and with at most $4n-9$ bends in total; see  Fig.~\ref{fi:example-2} for an illustration. 

For space reasons some proofs are in appendix.


\section{Preliminaries}\label{sec:preliminaries}

We assume familiarity with common notation and definitions about graphs, drawings, and planarity (see, e.g.,~\cite{book}). 
%
\begin{figure}[t]
\centering
\begin{subfigure}{0.49\textwidth}
	\centering
	\includegraphics[width=0.6\textwidth,page=1]{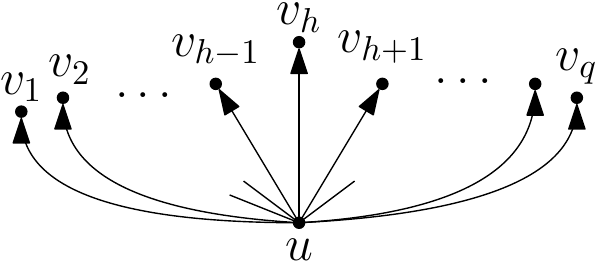}
	\caption{$\dots$$<$$\bo(v_{h-1})$$<$$\bo(v_{h})$$>$$\bo(v_{h+1})$$>$$\dots$}
	\label{fi:bitonic}
\end{subfigure}
\hfil
\begin{subfigure}{0.49\textwidth}
	\centering
	\includegraphics[width=0.6\textwidth,page=2]{figs/configurations}
	\caption{$(\bo(v_{i})$$>$$\bo(v_{i+1}))~\text{and}~(\bo(v_{j})$$<$$\bo(v_{j+1}))$}
	\label{fi:forbidden}
\end{subfigure}
\caption{(a) A bitonic sequence. (b) A forbidden configuration. }
\end{figure} 
An \emph{upward planar drawing} of a directed simple graph (or \emph{digraph} for short) $G$ is a planar drawing such that each edge of $G$ is drawn as a curve monotonically non-decreasing in the $y$-direction. An upward drawing is \emph{strict} if its edge curves are monotonically increasing. A digraph is \emph{upward planar} if it admits an upward planar drawing. Note that if a digraph admits an upward drawing then it also admits a strict upward drawing.  A digraph is upward planar if and only if it is a subgraph of a \emph{planar $st$-graph}~\cite{DBLP:journals/tcs/BattistaT88}. Let $G=(V,E)$ be an $n$-vertex planar $st$-graph, i.e., $G$ is a plane acyclic digraph with a single source $s$ and a single sink $t$, such that $s$ and $t$ belong to the boundary of the outer face and the edge $(s,t) \in E$~\cite{DBLP:journals/tcs/BattistaT88}. (Other works do not explicitly require the edge $(s,t)$ to be part of $G$, see, e.g.,~\cite{DBLP:conf/gd/Gronemann16}.) An \emph{$st$-ordering} of $G$ is a numbering $\bo: V \rightarrow \{1,2,\dots,n\}$ such that for each edge $(u,v) \in E$, it holds $\bo(u) < \bo(v)$ (which implies $\bo(s)=1$ and $\bo(t)=n$). Every planar $st$-graph has an $st$-ordering, which can be computed in $O(n)$ time (see, e.g.,~\cite{DBLP:books/daglib/0023376}).
If $u$ and $v$ are two adjacent vertices of $G$ such that $\bo(u)<\bo(v)$, we say that $v$ is a \emph{successor} of $u$, and $u$ is a \emph{predecessor} of $v$. Denote by $S(u) = \{v_1,v_2,\dots,v_q\}$ the sequence of successors of $v$ ordered according to the clockwise circular order of the edges incident to $u$ in the planar embedding of $G$. The sequence $S(u)$ is \emph{bitonic} if there exists an integer $1 \le h \le q$ such that $\bo(v_1) <  \dots < \bo(v_{h-1}) < \bo(v_h) > \bo(v_{h+1}) > \dots > \bo(v_q)$; see Fig.~\ref{fi:bitonic} for an illustration. Notice that when $h=1$ or $h=q$, $S(u)$ is actually a monotonic decreasing or increasing sequence. A \emph{bitonic $st$-ordering} of $G$ is an $st$-ordering such that, for every vertex $u \in V$, $S(u)$ is bitonic~\cite{DBLP:conf/gd/Gronemann16}. A planar $st$-graph $G$ is a \emph{bitonic $st$-graph} if it admits a bitonic $st$-ordering. Deciding whether $G$ is bitonic can be done in linear time both in the fixed~\cite{DBLP:conf/gd/Gronemann16} and in the variable~\cite{DBLP:conf/gd/ChaplickCCLNPTW17} embedding settings. If $G$ is not bitonic, every $st$-ordering $\bo$ of $G$ contains a forbidden configuration defined as follows. A sequence of successors $S(u)$ of a vertex $u$ forms a \emph{forbidden configuration} if there exist two indices $i$ and $j$, with $i<j$, such that $\bo(v_i) > \bo(v_{i+1})$ and $\bo(v_j) < \bo(v_{j+1})$, i.e. there is a path from $v_{i+1}$ to $v_i$ and a path from $v_j$ to $v_{j+1}$; see Fig.~\ref{fi:forbidden}.

Let $G=(V,E)$ be an $n$-vertex maximal plane graph with vertices $u$, $v$, and $w$ on the boundary of the outer~face. A \emph{canonical ordering}~\cite{DBLP:journals/combinatorica/FraysseixPP90} of $G$ is a linear ordering $\co = \{v_1=u,v_2=v,\dots,v_n=w\}$ of $V$, such that for every $3 \leq i \leq n$:
\begin{inparaenum}[\bfseries{C}1:]
\item\label{p:co1} The subgraph $G_i$ induced by $\{v_1,v_2,\dots,v_i\}$ is $2$-connected and internally triangulated, while the boundary of its outer face $C_i$ is a cycle containing $(v_1,v_2)$;
\item\label{p:co2} If $i+1\leq n$, $v_{i+1}$ belongs to $C_{i+1}$ and its neighbors in $G_{i}$ form a subpath of the path obtained by removing $(v_1,v_2)$ from $C_i$.
\end{inparaenum}

Computing $\co$ takes $O(n)$ time~\cite{DBLP:journals/combinatorica/FraysseixPP90}. Also, $\co$ is \emph{upward} if for every edge $(u,v)$ of a digraph $G$ $u$ precedes~$v$~in~$\co$.

The \emph{slope} of a line $\ell$ is the angle $\alpha$ that a horizontal line needs to be rotated counter-clockwise in order to make it overlap with $\ell$. If $\alpha=0$ we say that the slope of $\ell$ is \emph{horizontal}. The \emph{slope} of a segment is the slope of the line containing it. Let $\Ss=\{\alpha_1,\dots,\alpha_h\}$ be a set of $h$ slopes such that $\alpha_i < \alpha_{i+1}$. The slope set $\Ss$ is \emph{equispaced} if $\alpha_{i+1} - \alpha_i = \frac{\pi}{h}$, for $i=1,\dots,h-1$.  
Consider a \emph{$k$-bend planar drawing} $\Gamma$ of a graph $G$, i.e., a planar drawing in which every edge is mapped to a polyline containing at most $k+1$ segments. For a vertex $v$ in $\Gamma$ each slope $\alpha \in \Ss$ defines two different rays that emanate from $v$ and have slope $\alpha$. If $\alpha$ is horizontal these rays are called \emph{left horizontal} ray and \emph{right horizontal} ray. Otherwise, one of them is the \emph{top} and the other one is the \emph{bottom} ray of $v$. We say that a ray $r_v$ of a vertex $v$ is \emph{free} if there is no edge attached to $v$ through $r_v$ in $\Gamma$. We also say that $r_v$ is \emph{outer} if it is free and the first face encountered when moving from $v$ along $r_v$ is the outer face of $\Gamma$. 
The \emph{slope number} of a $k$-bend drawing $\Gamma$ is the number of distinct slopes used for the edge segments of $\Gamma$. The \emph{$k$-bend upward planar slope number} of an upward planar digraph $G$ is the minimum slope number over all $k$-bend upward planar~drawings~of~$G$. 


\section{$1$-bend Upward Planar Drawings}\label{sec:1bend}
  
Let $G=(V,E)$ be an $n$-vertex planar $st$-graph with a bitonic $st$-ordering $\bo=\{v_1,v_2,\dots,v_n\}$; see, e.g., Fig.~\ref{fi:bitonic-1}. We begin by describing an augmentation technique to ``transform'' $\bo$ into an upward canonical ordering of a suitable supergraph $\auG$ of $G$. We start from a result by Gronemann~\cite{DBLP:conf/gd/Gronemann16}, whose properties are summarized in the following lemma; see, e.g., Fig.~\ref{fi:bitonic-2}.

\begin{lemma}[\cite{DBLP:conf/gd/Gronemann16}]\label{le:gronemann}
	Let $G=(V,E)$ be an $n$-vertex planar $st$-graph that admits a bitonic $st$-ordering $\bo=\{v_1,v_2,\dots,v_n\}$. There exists a planar $st$-graph $G'=(V',E')$ with an $st$-ordering $\co=\{v_L,v_R,v_1,v_2,\dots,v_n\}$ such that:
\begin{inparaenum}[(i)]
\item $V'=V\cup\{v_L,v_R\}$; 
\item $E \subset E'$ and $(v_L,v_R) \in E'$; 
\item $v_L$ and $v_R$ are on the boundary of the outer face of $G'$; 
\item\label{p:predecessors} Every vertex of $G$ with less than two predecessors in $\bo$ has exactly two predecessors in $\co$.
\end{inparaenum}
\noindent Also, $G'$ and $\co$ are computed in $O(n)$~time.
\end{lemma}
We call $G'$ a \emph{canonical augmentation} of $G$. Observe that $G'$ always contains the edges $(v_L,v_1)$ and $(v_R,v_1)$ because of~(\ref{p:predecessors}). We also insert the edge $(v_L,v_n)$, which is required according to our definition of $st$-graph; this addition is always possible because $v_L$ and $v_n$ are both on the boundary of the outer face. The next lemma shows that any planar $st$-graph obtained by triangulating $G'$ admits an upward canonical ordering; see, e.g., Fig.~\ref{fi:bitonic-3}.  

\begin{figure}[t]
\centering
\begin{subfigure}{0.32\textwidth}
	\centering
	\includegraphics[width=\textwidth,page=1]{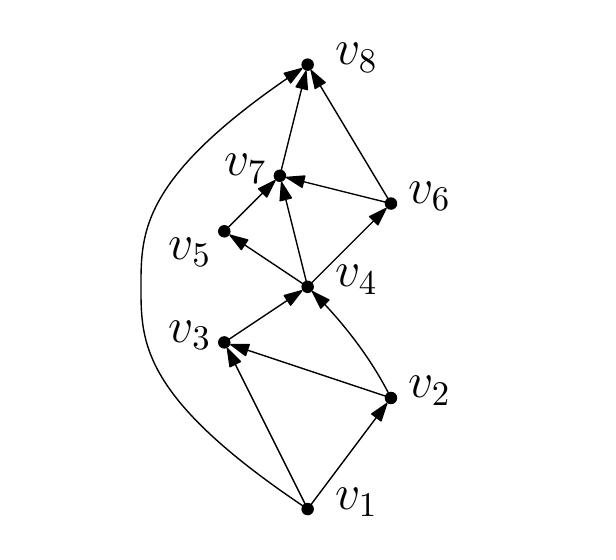}
	\caption{}
	\label{fi:bitonic-1}
\end{subfigure}
\begin{subfigure}{0.32\textwidth}
	\centering
	\includegraphics[width=\textwidth,page=2]{figs/bitonic}
	\caption{}
	\label{fi:bitonic-2}
\end{subfigure}
\begin{subfigure}{0.32\textwidth}
	\centering
	\includegraphics[width=\textwidth,page=3]{figs/bitonic}
	\caption{}
	\label{fi:bitonic-3}
\end{subfigure}
\caption{(a) A bitonic $st$-graph $G$ with $\bo=\{v_1,v_2,\dots,v_8\}$. (b) A canonical augmentation $G'$ of $G$ with $\co=\{v_L,v_R,v_1,v_2,\dots,v_8\}$. (c) A planar $st$-graph $\auG$ obtained by triangulating $G'$. $\co$ is an upward canonical ordering of $\auG$. }
\end{figure}

\begin{lemma}\label{le:augmentation}
Let $G'$ be a canonical augmentation of an $n$-vertex bitonic $st$-graph $G$. Every planar $st$-graph $\auG$ obtained by triangulating $G'$ has the following properties:
\begin{inparaenum}[(a)]
\item\label{p:simpletr} it has no parallel edges;
\item\label{p:upwardco} $\co=\{v_L,v_R,v_1,v_2,\dots,v_n\}$ is an upward canonical ordering. 
\end{inparaenum} 
\end{lemma}
\begin{proof}
Concerning Property~(\ref{p:simpletr}), suppose for a contradiction that $\auG$ has two parallel edges $e_1$ and~$e_2$ connecting $u$ with~$v$. Let $\mathcal{C}$ be the $2$-cycle formed by $e_1$ and $e_2$ and let $V_\mathcal{C}$ be the set of vertices distinct from $u$ and $v$ that are inside $\mathcal C$ in the embedding of $\auG$. $V_\mathcal{C}$ is not empty, as otherwise $\mathcal C$ would be a non-triangular face of $\auG$. Let $w$ be the vertex with the lowest number in $\co$ among those in $V_\mathcal{C}$. Since $\auG$ is planar (in particular $e_1$ and $e_2$ are not crossed) and has a single source, it contains a directed path from $u$ to every vertex in $V_\mathcal{C}$. Hence, it has an edge from $u$ to $w$. Also, by assumption, there is no vertex $z$ in $V_\mathcal{C}$ such that $\co(z) < \co(w)$, which implies that $u$ is the only predecessor of $w$ in $\co$, a contradiction to Lemma~\ref{le:gronemann}(\ref{p:predecessors}). Concerning Property~(\ref{p:upwardco}), if $\co$ is a canonical ordering of $\auG$, then $\co$ is actually an upward canonical ordering because it is also an $st$-ordering. To see that $\co$ is a canonical ordering, observe first that $v_L$, $v_R$ and $v_n$ are on the boundary of the outer face of $\auG$ by construction. Denote by $\auG_i$ the subgraph of $\auG$ induced by $\{v_L,v_R,v_1,\dots,v_i\}$ and let $\auC_i$ be the boundary of its outer face. We first prove by induction on $i$ (for $i=1,2,\dots,n$) that $\auG_i$ is $2$-connected. In the base case $i=1$, $\auG_1$ is a $3$-cycle and therefore it is $2$-connected. In the case $i>1$, $\auG_{i-1}$ is $2$-connected by induction and $v_{i}$ has at least two predecessors in $\auG_{i-1}$ by Lemma~\ref{le:gronemann}(\ref{p:predecessors}), thus $\auG_{i}$ is $2$-connected. We now prove that each $\auG_i$, for $i=1,2,\dots,n$, is internally triangulated, which concludes the proof of condition \myref{C\ref{p:co1}} of canonical ordering. Suppose, for a contradiction, that there exists an inner face $f$ that is not a triangle. Since $\auG$ is triangulated, there exists a vertex $v_j$, with $j > i$, that is embedded inside $f$ in $\auG_j$. Since $\co$ is an $st$-ordering, there is no directed path from $v_{j}$ to any vertex of $f$. On the other hand, either $v_{j}=v_n$ or there is a directed path from $v_{j}$ to $v_n$. Both cases contradict the fact that $v_n$ belongs to the boundary of the outer face of $\auG$. We finally show that $v_{i}$ belongs to $C_{i}$, for $i=1,2,\dots,n$. Since we already proved that $\auG_{i}$ is triangulated, this is enough to prove \myref{C\ref{p:co2}}. By the planarity of $\auG_i$, there is a face $f$ in $\auG_{i-1}$ such that all the neighbors of $v_{i}$ in $\auG_{i-1}$ belong to the boundary of $f$. We claim that $f$ is the outer face of $\auG_{i}$. If it was an inner face, then $v_{i}$ would be embedded inside $f$ in $\auG_{i}$ and, by the same argument used above,  $v_n$ would not belong to the boundary of the outer face of $\auG$.
\end{proof}

We now show that any set of $\Delta$ slopes $\Ss$ that contains the horizontal slope is universal for $1$-bend upward planar drawings of bitonic $st$-graphs. The algorithm is inspired by a technique of Angelini et al.~\cite{DBLP:conf/compgeom/AngeliniBLM17}. We will use important additional tools with respect to~\cite{DBLP:conf/compgeom/AngeliniBLM17}, such as the construction of a triangulated canonical augmentation,  extra slopes to draw the edges inserted by the augmentation procedure, and different geometric invariants.
Let $G$ be an $n$-vertex bitonic $st$-graph with maximum vertex degree $\Delta$; see Fig.~\ref{fi:bitonic-1}. The algorithm first computes a triangulated canonical augmentation $\auG$ of $G$; see Figs.~\ref{fi:bitonic-2}--\ref{fi:bitonic-3}. We call \emph{dummy edges} all edges that are in $\auG$ but not in $G$ and \emph{real edges} the edges in $\auG$ that are also in $G$. By~Lemma~\ref{le:augmentation}, $\auG$ admits an upward canonical ordering $\co=\{v_L,v_R,v_1,v_2,\dots,v_n\}$, where $\co$ is an $st$-ordering such that each vertex distinct from $v_L$ and $v_R$ has at least two predecessors. Let $\Ss=\{\rho_1,\dots,\rho_\Delta\}$ be any set of $\Delta$ slopes, which we call \emph{real slopes}. Let $\rho^*$ be the smallest angle between two slopes in $\Ss$ and let $\Delta^*$ be the maximum number of dummy edges incident to a vertex of $\auG$. For each slope $\rho_i$ ($1 \le i \le \Delta)$, we add $\Delta^*$ \emph{dummy slopes} $\{\delta^i_1,\dots,\delta^i_{\Delta^*}\}$ such that $\delta^i_j = \rho_i + j \cdot \frac{\rho^*}{\Delta^*+1}$, for $j=1,2,\dots, \Delta^*$. Hence, there are $\Delta^*$ dummy slopes between any two consecutive real slopes. We will use the real slopes for the real edges and the dummy slopes for the dummy ones. 

Let $\auG_i$ be the subgraph of $\auG$ induced by $\{v_L,v_R,v_1,v_2,\dots,v_i\}$. The algorithm  constructs the drawing by adding the vertices according to $\co$. More precisely, it computes a drawing $\auGamma_i$ of the digraph $\auG^-_i$ obtained from $\auG_i$ by removing the dummy edges  $(v_L,v_R)$ and $(v_1,v_R)$, which exist by construction, and $(v_R,v_2)$ if it exists. Let $\auC_i$ be the boundary of the outer face of $\auG_i$, and let $\auP_i$ be the path obtained by removing $(v_L,v_R)$ from $\auC_i$.
For a vertex $v$ of $\auP_i$, we denote by $\nr{v}{i}$ (resp.\ $\nd{v}{i}$) the number of real (resp.\ dummy) edges incident to $v$ that are not in $\auG_{i}$ and by $\bcw{j}(v,i)$ (resp.\ $\bccw{j}(v,i)$) the $j$-th outer real top ray in $\auGamma_{i}$ encountered in clockwise (resp.\ counterclockwise) order around $v$ starting from the left (resp.\ right) horizontal ray.
For dummy top rays, we define analogously $\rcw{j}(v,i)$ and $\rccw{j}(v,i)$.
\begin{figure}[t]
\centering
\begin{subfigure}{0.49\textwidth}
	\centering
	\includegraphics[width=\textwidth,page=1]{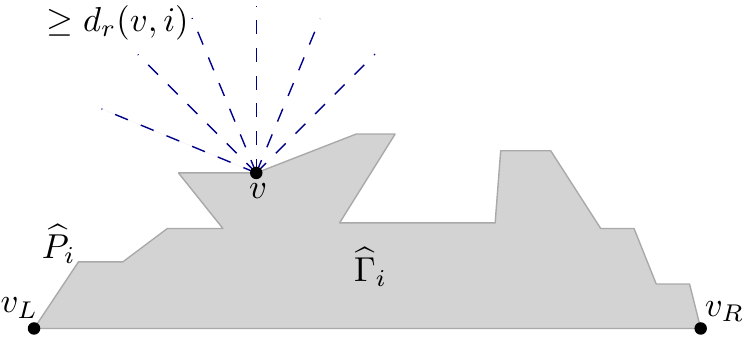}
	\caption{\myref{I3}}
	\label{fi:invariants-i3}
\end{subfigure}
\hfil
\begin{subfigure}{0.49\textwidth}
	\centering
	\includegraphics[width=\textwidth,page=2]{figs/invariants}
	\caption{\myref{I4}--\myref{I5}}
	\label{fi:invariants-i4-i5}
\end{subfigure}
	\begin{subfigure}{0.32\textwidth}
		\centering
		\includegraphics[width=\textwidth,page=1]{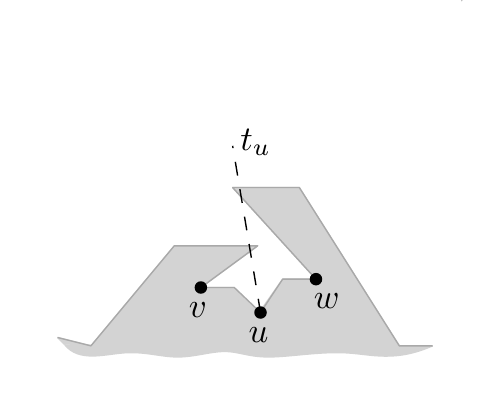}
		\caption{}
		\label{fi:stretch1}
	\end{subfigure}
	\hfil
	\begin{subfigure}{0.32\textwidth}
		\centering
		\includegraphics[width=\textwidth,page=2]{figs/stretch}
		\caption{}
		\label{fi:stretch2}
	\end{subfigure}
\caption{(a)-(b) Illustration for Invariants \myref{I3}--\myref{I5}; real rays are dashed, dummy rays are dotted. (c)-(d) Illustration for Lemma~\ref{le:stretching2}.}
\end{figure}
$\auGamma_i$  satisfies~the~following~{invariants}:
\begin{description}
	\item[{I1}] $\auGamma_i$ is a $1$-bend upward planar drawing whose real edges use only slopes in $\Ss$.
	
	\item[{I2}] Every edge of $\auP_i$ contains a horizontal segment.
	
	\item[{I3}] Every vertex $v$ of $\auP_i$ has at least $\nr{v}{i}$ outer real top rays; see Fig.~\ref{fi:invariants-i3}.
	
	\item[{I4}] Every vertex $v$ of $\auP_i$ has at least $\nd{v}{i}$ outer dummy top rays between $\rcw{1}(v,i)$ and $\bcw{1}(v,i)$ (resp.\ $\rccw{1}(v,i)$ and $\bccw{1}(v,i)$), including $\rcw{1}(v,i)$ (resp.\ $\rccw{1}(v,i)$); see Fig.~\ref{fi:invariants-i4-i5}.   
	\item[{I5}] Let $\ell$ be any horizontal line and let $p$ and $p'$ be any two intersection points between $\ell$ and the polyline representing $\auP_i$ in $\auGamma_{i}$; walking along $\ell$  from left to right, $p$ and $p'$ are encountered in the same order as when walking along $\auP_i$ from $v_L$ to $v_R$; see Fig.~\ref{fi:invariants-i4-i5}.
\end{description}

The last vertex $v_n$ is added to $\auGamma_{n-1}$ in a slightly different way and the resulting drawing will satisfy \myref{I1}. 
The next two lemmas state  important properties of any $1$-bend upward planar drawing satisfying \myref{I1}--\myref{I5}. Similar lemmas are proven in~\cite[Lemmas 2 and 3]{DBLP:conf/compgeom/AngeliniBLM17}, but for drawings that satisfy different invariants.

\begin{lemma}\label{le:stretching}
	Let $\auGamma_i$ be a drawing of $\auG^-_i$ that satisfies Invariants \myref{I1}--\myref{I5}. Let $(u,v)$ be any edge of $\auP_i$ such that $u$ is encountered before $v$  along $\auP_i$ when going from $v_L$ to $v_R$, and let $\lambda$ be a positive number. There exists a drawing $\auGamma'_i$ of $\auG^-_i$ that satisfies Invariants \myref{I1}--\myref{I5} and such that: (i) the horizontal distance between $u$ and $v$ is increased by $\lambda$; (ii) the horizontal distance between any two other consecutive vertices along $\auP_i$ is the same as in $\auGamma_i$.
\end{lemma}

\noindent The next lemma can be proven by suitably applying Lemma~\ref{le:stretching}; see Figs.~\ref{fi:stretch1}-\ref{fi:stretch2}.

\begin{lemma}\label{le:stretching2}
	Let $\auGamma_i$ be a drawing of $\auG^-_i$ that satisfies Invariants \myref{I1}--\myref{I5}. Let $u$ be a vertex of $\auP_i$, and let $t_u$ be any outer top ray of $u$ that crosses an edge of $\auG^-_i$ in $\auGamma_i$. There exists a drawing $\auGamma'_i$ of $\auG^-_i$ that satisfies Invariants \myref{I1}--\myref{I5} in which $t_u$ does not cross any edge of $\auG^-_i$.
\end{lemma}


We now describe our drawing algorithm starting with the computation of $\auGamma_2$. We aim at drawing both $v_1$ and $v_2$ horizontally aligned between $v_L$ and $v_R$. Note that $v_1$ is the source of $G$, and, by the definition of a canonical augmentation, $v_1$ is adjacent to both $v_L$ and $v_R$, while $v_2$ is adjacent to $v_1$ and to at least one of $v_L$ and $v_R$. We remove the dummy edges $(v_1,v_R)$ and $(v_L,v_R)$, and the dummy edge $(v_R,v_2)$ if it exists. The resulting graph is either the path~$\langle v_L, v_1, v_2, v_R \rangle$ or the path~$\langle v_L, v_2, v_1, v_R \rangle$, which we draw along~a~horizontal~segment. 

\begin{lemma}\label{le:gamma2}
	Drawing $\auGamma_2$ satisfies Invariants \myref{I1}--\myref{I5}.
\end{lemma}

Assume now that we have constructed drawing $\auGamma_{i-1}$ of $\auG_{i-1}$ satisfying \myref{I1}--\myref{I5} $(3 \leq i < n)$. Let $\{u_1,\dots,u_q\}$ be the neighbors of the next vertex $v_i$ along $\auP_{i-1}$. Let $t_{1}$ be either $\bccw{1}(u_1,i-1)$, if $(u_1,v_i)$ is real, or $\rccw{1}(u_1,i-1)$, if $(u_1,v_i)$ is dummy. Symmetrically, let $t_{q}$ be either $\bcw{1}(u_q,i-1)$, if $(u_q,v_i)$ is real, or $\rcw{1}(u_q,i-1)$, if $(u_q,v_i)$ is dummy. 
Let $t_{j}$ (for $1 < j < q$) be any outer real (resp.\ dummy) top ray emanating from $u_j$ if $(u_j,v_i)$ is real (resp.\ dummy). By  \myref{I3} all such top rays exist and by Lemma~\ref{le:stretching2} we can assume that none of them crosses~$\auGamma_{i-1}$.
\begin{figure}[t]
	\centering
	\begin{subfigure}{0.32\textwidth}
		\centering
		\includegraphics[width=\textwidth,page=1]{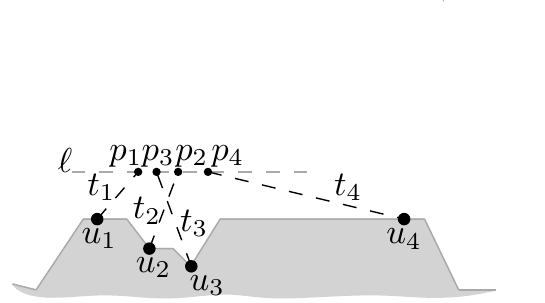}
		\caption{}
		\label{fi:vertex-addition1}
	\end{subfigure}
	\hfil
	\begin{subfigure}{0.32\textwidth}
		\centering
		\includegraphics[width=\textwidth,page=2]{figs/vertex-addition}
		\caption{}
		\label{fi:vertex-addition2}
	\end{subfigure}
	\begin{subfigure}{0.32\textwidth}
		\centering
		\includegraphics[width=\textwidth,page=3]{figs/vertex-addition}
		\caption{}
		\label{fi:vertex-addition3}
	\end{subfigure}
	\begin{subfigure}{0.48\textwidth}
		\centering
		\includegraphics[width=\textwidth,page=5]{figs/vertex-addition}
		\caption{}
		\label{fi:vertex-addition4}
	\end{subfigure}
	\begin{subfigure}{0.48\textwidth}
		\centering
		\includegraphics[width=\textwidth,page=6]{figs/vertex-addition}
		\caption{}
		\label{fi:vertex-addition5}
	\end{subfigure}
	\caption{Addition of vertex $v_i$.}
\end{figure}
Let $\ell$ be a horizontal line above the topmost point of $\auGamma_{i-1}$. Let $p_j$ be the intersection point of $t_{j}$ and $\ell$. We can assume that, for $j=1,2,\dots,q-1$, $p_j$ is to the left of $p_{j+1}$. If this is not the case, we can increase the distance between $u_j$ and $u_{j+1}$ so to guarantee that $p_j$ and $p_{j+1}$ appear in the desired order along $\ell$; this can be done by applying Lemma~\ref{le:stretching} with respect to each edge $(u_j,u_{j+1})$ for a suitable choice of $\lambda$; see Figs.~\ref{fi:vertex-addition1}-\ref{fi:vertex-addition2} for an illustration.
We will place $v_i$ above $\ell$ using $q-2$ bottom rays $b_2,b_3,\dots,b_{q-1}$  of $v_i$ for the segments of the edges $(u_j,v_i)$ ($j=2,3,\dots,q-1$) incident to $v_i$ such that: (i) $b_j$ ($1 <j < q$) is real (resp.\ dummy) if $(u_j,v_i)$ is real (resp.\ dummy); (ii) $b_j$ precedes $b_{j+1}$ in the counterclockwise order around $v_i$ starting from $b_2$. This choice is possible for the real rays  because $v_i$ has $\Delta-1$ real bottom  rays and it has at least one incident real edge not in $\auG_{i}$ (otherwise it would be a sink of $G$, which is not possible because $i < n$). Concerning the dummy rays, we have at most $\Delta^*$ dummy edges incident to $v_i$ and $\Delta^*$ dummy bottom rays between any two consecutive real rays.    
Consider the ray $t_{1}$ and choose a point $p$ to the right of $t_{1}$ and above $\ell$ such that placing $v_i$ on $p$ guarantees that $\min_{i=1 \dots q-2}\{x(p'_{i+1})-x(p'_i)\} > x(p_q)-x(p_1)$, where $p'_1=p_1$ and  $p'_2,p'_3,\dots,p'_{q-1}$ are the intersection points of the rays $b_2,b_3,\dots,b_{q-1}$ with the line $\ell$ (see Fig.~\ref{fi:vertex-addition3}). Observe that for a sufficiently large $y$-coordinate, point $p$ can always be found. 
We now apply Lemma~\ref{le:stretching} to each of the edges $(u_1,u_2)$, $(u_2,u_3)$, $\dots$, $(u_{q-2},u_{q-1})$, in this order, choosing $\lambda \geq 0$ so that each $p_j$ is translated to $p'_j$ (for $j=2,3,\dots,q-1$). We finally apply again the same procedure to $(u_{q-1},u_q)$ so that the intersection point between $t_{q}$ and the horizontal line $\ell_H$ passing through $v_i$ is to the right of $v_i$ (see Fig.~\ref{fi:vertex-addition4}). After this translation procedure, we can draw the edge $(u_1,v_i)$ (resp.\ $(u_q,v_i)$) with a bend at the intersection point between $t_{1}$ (resp.\ $t_{q}$) and $\ell_H$ and therefore using the slope of $t_{1}$ (resp.\ $t_{q}$) and the horizontal slope (see Fig.~\ref{fi:vertex-addition5}). The edges $(u_j,v_i)$ ($j=2,3,\dots,q-1$) are drawn with a bend point at $p_j=p'_j$ and therefore using the slopes of $t_{j}$ and $b_j$.

\begin{lemma}\label{le:gammai}
	Drawing $\auGamma_i$, for $i=3,4,\dots,n-1$, satisfies Invariants \myref{I1}--\myref{I5}.
\end{lemma}
\begin{proof}
	The proof is by induction on $i \geq 3$. $\auGamma_{i-1}$ satisfies Invariants \myref{I1}--\myref{I5} by Lemma~\ref{le:gamma2} when $i=3$, and by induction when $i > 3$. 
	
	\noindent \textbf{\myref{Proof of I1}}. By construction, each $(u_j,v_i)$ ($j=1,2,\dots,q$) is drawn as a chain of at most two segments that use real and dummy slopes. In particular, if $(u_j,v_i)$  is real, then it uses real slopes, i.e., slopes in $\Ss$. By the choice of $\ell$, the bend point of $(u_j,v_i)$ has $y$-coordinate strictly greater than that of $u_j$ and smaller than or equal to that of $v_i$. Since each $(u_j,v_i)$ is oriented from $u_j$ to $v_i$ (as $\co$ is an upward canonical ordering), the drawing is upward. Concerning planarity, we first observe that $\auGamma_{i-1}$ is planar and it remains planar each time we apply Lemma~\ref{le:stretching}. Also, by Lemma~\ref{le:stretching2} each $(u_j,v_i)$ ($j=1,2,\dots,q$) does not intersect $\auGamma_{i-1}$ (except at $u_j$). Further, the order of the bend points along $\ell$ guarantees that the edges incident to $v_i$ do not cross each~other. 
	
\noindent \textbf{\myref{Proof of I2}}. The only edges of $\auP_i$ that are not in $\auP_{i-1}$ are $(u_1,v_i)$ and $(u_q,v_i)$. For both these edges the segment incident to $v_i$ is horizontal by construction.
	
\noindent \textbf{\myref{Proof of I3}}. For each vertex of $\auP_i$ distinct from $u_1$, $u_q$ and $v_i$, \myref{I3} holds by induction. Invariant \myref{I3} also holds for $v_i$ because $\nr{v_i}{i} \leq \Delta-1$ (as otherwise $v_i$ would be a source of $G$, which is not possible because $i>1$) and all the real top rays of $v_i$, which are $\Delta-1$, are outer. Consider now vertex $u_1$ (a symmetric argument applies to $u_q$). If $(u_1,v_i)$ is real, then $\nr{u_1}{i}=\nr{u_1}{i-1}-1$; in this case $t_{1}=\bccw{1}(u_1,i-1)$ and therefore all the other $\nr{u_1}{i-1}-1$ outer real top rays of $u_1$ in $\auGamma_{i-1}$ remain outer in $\auGamma_{i}$. If $(u_1,v_i)$ is dummy, then $\nr{u_1}{i}=\nr{u_1}{i-1}$; in this case $t_{1}=\rccw{1}(u_1,i-1)$ and therefore all the $\nr{u_1}{i-1}$ outer real top rays of $u_1$ in $\auGamma_{i-1}$ remain outer in $\auGamma_{i}$.

\noindent \textbf{\myref{Proof of I4}}. For each vertex of $\auP_i$ distinct from $u_1$, $u_q$ and $v_i$, \myref{I4} holds by induction. \myref{I4} also holds for $v_i$ because $\nd{v_i}{i} \leq \Delta^*$ and there are $\Delta^*$ dummy top rays between $\rcw{1}(v_i,i)$ and $\bcw{1}(v_i,i)$ including $\rcw{1}(v_i,i)$ (all the top rays of $v_i$ are outer). Analogously, there are $\Delta^*$ outer dummy top rays between $\rccw{1}(v_i,i)$ and $\bccw{1}(v_i,i)$ including $\rccw{1}(v_i,i)$. Consider now $u_1$ (a symmetric argument applies to $u_q$). If $(u_1,v_i)$ is real, then $\nd{u_1}{i}=\nd{u_1}{i-1}$; in this case $t_{1}=\bccw{1}(u_1,i-1)$ and there are $\Delta^*$ outer dummy top rays between $\rccw{1}(u_1,i)$ and $\bccw{1}(u_1,i)$ including $\rccw{1}(u_1,i)$ (namely, all those between $t_{1}=\bccw{1}(u_1,i-1)$ and $\bccw{2}(u_1,i-1)$). If $(u_1,v_i)$ is dummy, then $\nd{u_1}{i}=\nd{u_1}{i-1}-1$; in this case $t_{1}=\rccw{1}(u_1,i-1)$ and therefore all the other $\nd{u_1}{i-1}-1$ outer dummy top rays of $u_1$, which by induction were between $\rccw{1}(u_1,i-1)$ and $\bccw{1}(u_1,i-1)$, remain outer in $\auGamma_{i}$. 	 
	
 \noindent \textbf{\myref{Proof of I5}}. Notice that the various applications of Lemma~\ref{le:stretching} to $\auGamma_{i-1}$ preserve~\myref{I5}. Let $p$ and $p'$ be any two intersection points between a horizontal line $\ell$ and the polyline representing $\auP_i$ in $\auGamma_{i}$, with $p$ to the left of $p'$ along $\ell$. If $p$ and $p'$ belong to $\auP_{i-1}$, \myref{I5} holds by induction. If both $p$ and $p'$ belong to the path $\langle u_1, v_i, u_q \rangle$, \myref{I5} holds by construction. If $p$ belongs to $\auP_{i-1}$ and $p'$ belongs to $\langle u_1, v_i, u_q \rangle$, then $p$ belongs to the subpath of $\auP_{i-1}$ that goes from $v_L$ to $u_1$ because the subpath from $u_q$ to $v_R$ is completely to the right of $t_{q}$, hence \myref{I5} holds also in this case. If $p$ belongs to $\langle u_1, v_i, u_q \rangle$ and $p'$ belongs to $\auP_{i-1}$, the proof is symmetric.    
\end{proof}

\begin{lemma}\label{le:bitonic-correct}
	$G$ has a $1$-bend upward planar drawing $\Gamma$ using only slopes in $\Ss$.
\end{lemma}
\begin{proof}
	By Lemma~\ref{le:gammai}, drawing $\auGamma_{n-1}$ satisfies Invariant \myref{I1}--\myref{I5}. We explain how to add the last vertex $v_n$ to obtain a drawing that satisfies Invariant \myref{I1}. Let  $\{u_1,\dots,u_q\}$ be the predecessors of $v_n$ on $\auP_{n-1}$. Notice that, in this case $u_1=v_L$ and $u_q=v_R$. Vertex $v_n$ is added to the drawing similarly to all the other vertices added in the previous steps of the algorithm. The only difference is that the number of real incoming edges incident to $v_n$ in $\auGamma_{n-1}$ can be up to $\Delta$. If this is the case, since the real bottom rays are $\Delta-1$, they are not enough to draw all the real edges incident to $v_n$. Let $j$ be the smallest index such that $(u_j,v_n)$ is a real edge. We ignore all the dummy edges $(u_h,v_n)$, for $h=1,2,\dots,j-1$, and apply the construction used in the previous steps considering only $\{u_j,u_{j+1},\dots,u_q\}$ as predecessors of $v_n$ (notice that such predecessors are at least two because $v_n$ has at least two incident real edges). By ignoring these dummy edges, the segment of the real edge $(u_j,v_n)$ incident to $v_n$ will be drawn using the left horizontal slope. Denote by $\auGamma_{n}$ the resulting drawing. As in the proof of Lemma~\ref{le:gammai}, we can prove that \myref{I1} holds for $\auGamma_{n}$ and therefore $\auGamma_{n}$ is a $1$-bend upward planar drawing whose real edges use only slopes in $\Ss$. The drawing $\Gamma$ of $G$ is obtained from $\auGamma_{n}$ by removing all its dummy edges and the two dummy vertices~$v_L$~and~$v_R$.
\end{proof}

\begin{lemma}\label{le:bitonic-time}
	Drawing $\Gamma$ can be computed in $O(n)$ time.
\end{lemma}
	
\noindent Lemmas~\ref{le:bitonic-correct} and~\ref{le:bitonic-time} are summarized by Theorem~\ref{thm:bitonic}. Corollary~\ref{co:1bendusn} is a consequence of Theorem~\ref{thm:bitonic} and of a result in~\cite{DBLP:conf/gd/GiacomoLM16}.

\begin{theorem}\label{thm:bitonic}
	Let $\Ss$ be any set of $\Delta \ge 2$ slopes including the horizontal slope and let $G$ be an $n$-vertex bitonic planar $st$-graph with maximum vertex degree $\Delta$. Graph $G$ has a $1$-bend upward planar drawing $\Gamma$ using only slopes in $\Ss$, which can be computed in $O(n)$ time.
\end{theorem}


\begin{corollary}\label{co:1bendusn}
	Every bitonic $st$-graph with maximum vertex degree $\Delta \geq 2$ has 1-bend upward planar slope number at most $\Delta$, which is worst-case optimal.
\end{corollary}

\noindent If $\Ss$ is equispaced, Theorem~\ref{thm:bitonic} implies a lower bound of $\frac{\pi}{\Delta}$ on the angular resolution of the computed drawing, which is worst-case optimal~\cite{DBLP:conf/gd/GiacomoLM16}. 
Also, Theorem~\ref{thm:bitonic} can be extended to planar $st$-graphs with $\Delta \le 3$, as any such digraph can be made bitonic by only  rerouting the edge $(s,t)$.

\begin{theorem}\label{thm:cubic-bitonic}
	Every planar $st$-graph with maximum vertex degree $3$ has $1$-bend upward planar slope number at most $3$.
\end{theorem}

We conclude with the observation that an upward drawing constructed by the algorithm of Theorem~\ref{thm:bitonic} can be transformed into a strict upward drawing that uses $\Delta+1$ slopes rather than $\Delta$. It suffices to replace every horizontal segment oriented from its leftmost (rightmost) endpoint to its rightmost (leftmost) one with a segment having slope $\varepsilon$ ($-\varepsilon$), for a sufficiently small value of $\varepsilon>0$.

\section{$2$-bend Upward Planar Drawings}\label{sec:2bend}


We now extend the result of Theorem~\ref{thm:bitonic} to non-bitonic planar $st$-graphs. By adapting a technique of Keszegh et al.~\cite{DBLP:journals/siamdm/KeszeghPP13}, one can construct $2$-bend upward planar drawings of planar $st$-graphs using at most $\Delta$ slopes. We improve upon this result in two ways: (i) The technique in~\cite{DBLP:journals/siamdm/KeszeghPP13} may lead to drawings with $5n-11$ bends in total, while we prove that $4n-9$ bends suffice; (ii) It uses a fixed set of $\Delta$ slopes (and it is not immediately clear whether it can work with any set of slopes), while we show that any set of $\Delta$ slopes with the~horizontal~one~is~universal. 

Let $G$ be an $n$-vertex non-bitonic planar $st$-graph. All forbidden configurations of $G$ can be removed in linear time by subdividing at most $n-3$ edges of $G$~\cite{DBLP:conf/gd/Gronemann16}. Let $G_b$ be the resulting bitonic $st$-graph, called a \emph{bitonic subdivision} of $G$. Let $\langle u,d,v \rangle$ be a directed path of $G_b$  obtained by subdividing the edge $(u,v)$ of $G$ with the dummy vertex $d$. We call $(u,d)$ the \emph{lower stub}, and $(d,v)$ the \emph{upper stub} of $(u,v)$. We can prove the existence of an augmentation technique similar to that of Lemma~\ref{le:gronemann}, but with an additional property on the upper stubs.

\begin{lemma}\label{le:gronemann2}
	Let $G=(V,E)$ be an $n$-vertex planar $st$-graph that is not bitonic. Let $G_b=(V_b,E_b)$ be an $N$-vertex bitonic subdivision of $G$, with a bitonic $st$-ordering $\bo=\{v_1,v_2,\dots,v_N\}$. There exists a planar $st$-graph $G'=(V',E')$ with an $st$-ordering $\co=\{v_L,v_R,v_1,v_2,\dots,v_N\}$ such that:
	\begin{inparaenum}[(i)]
		\item $V'=V_b\cup\{v_L,v_R\}$; 
	 	\item $E_b \subset E'$ and $(v_L,v_R) \in E'$; 
	 	\item $v_L$ and $v_R$ are on the boundary of the outer face of $G'$; 
	 	\item\label{p:predecessors2} Every vertex of $G_b$ with less than two predecessors in $\bo$ has exactly two predecessors in $\co$.
	 	\item\label{p:upperstub} There is no vertex in $G'$ such that its leftmost or its rightmost incoming edge is an upper stub. 
	\end{inparaenum}
	\noindent Also, $G'$ and $\co$ are computed~in~$O(n)$~time.
\end{lemma}

\begin{theorem}\label{thm:planarst}
Let $\Ss$ be any set of $\Delta \ge 2$ slopes including the horizontal slope and let $G$ be an $n$-vertex planar $st$-graph with maximum vertex degree $\Delta$. Graph $G$ has a $2$-bend upward planar drawing $\Gamma$ using only slopes in $\Ss$, which has at most $4n-9$ bends in total and which can be computed in $O(n)$ time.
\end{theorem}
\begin{proof}
	We compute a triangulated canonical augmentation $\auG$ of $G$ by (1) applying Lemma~\ref{le:gronemann2} and (2) triangulating the resulting digraph. By Lemma~\ref{le:augmentation}, $\auG$ has an upward canonical ordering $\co$.  
	The algorithm of Theorem~\ref{thm:bitonic} to $\auG$ would lead to a $3$-bend drawing of $G$ (by interpreting every subdivision vertex as a bend). We explain how to modify it to construct a drawing $\auGamma$ of $\auG$ with at most $2$ bends per edge and $4n-9$ bends in total. 
	Let $v_i$ the next vertex to be added according to $\co$ and let $\{u_1,u_2,\dots,u_q\}$ its neighbors in $\auP_{i-1}$. Suppose that $u_j$ is a dummy vertex and that $(u_j,v_i)$ is an upper stub. To save one bend along the edge subdivided by $u_j$, we draw $(u_j,v_i)$ without bends. By Lemma~\ref{le:gronemann2}~(\ref{p:upperstub}), we have that $1 < j < q$. The ray $t_{j}$ used to draw the segment of $(u_j,v_i)$ incident to $u_j$ can be any outer real top ray; we choose the ray with same slope as the real bottom  ray $b_j$ used to draw the segment of $(u_j,v_i)$ incident to $v_i$. This is possible because all real top rays of $u_j$ are outer (since $(u_j,v_i)$ is the only real outgoing edge of $u_j$). Hence, edge $(u_j,v_i)$ has no bends. 
	The drawing $\Gamma$ of $G$ is obtained from $\auGamma$ by removing dummy edges and replacing dummy vertices (except $v_L$ and $v_R$, which are removed) with bends. Since the upper stubs of subdivided edges has $0$ bends, each edge of $\Gamma$ has at most $2$ bends. Let $m_1$ and $m_2$ be the number of edges drawn with $1$ and $2$ bends, respectively; we have $m_2 \leq n-3$ and $m_1=m-m_2 \leq 3n-6 -(n-3)=2n-3$. Thus the total number of bends is at most $2n-3+2(n-3)=4n-9$. Finally, $\auG$ can be computed in $O(n)$ time (Lemma~\ref{le:gronemann2}) and the modified drawing algorithm still runs in linear time.
\end{proof}

A planar $st$-graph with a source/sink of degree $\Delta$ requires at least $\Delta -1$ slopes in any upward planar drawing; thus the gap with Theorem~\ref{thm:planarst} is one unit. Similarly to Theorem~\ref{thm:bitonic}, Theorem~\ref{thm:planarst} implies a lower bound of $\frac{\pi}{\Delta}$  on the angular resolution of $\Gamma$; an upper bound of $\frac{\pi}{\Delta-1}$ can be proven with the same digraph used for the lower bound on the~slope~number.
Finally, Theorem~\ref{thm:upward} extends the result of Theorem~\ref{thm:planarst} to every upward planar graph using an additional slope.

\begin{theorem}\label{thm:upward}
	Let $\Ss$ be any set of $\Delta+1$ slopes including the horizontal slope and let $G$ be an $n$-vertex upward planar graph with maximum vertex degree $\Delta \ge 2$. Graph $G$ has a $2$-bend upward planar drawing using only slopes in $\Ss$.
\end{theorem} 

\section{Open Problems}\label{sec:conclusions}

%
\begin{inparaenum}[(i)] 
	\item Can we draw every planar $st$-graph with at most one bend per edge (or less than $4n-9$ in total) and $\Delta$ slopes?
	\item What is the $2$-bend upward planar slope number of planar $st$-graphs? Is $\Delta$  a tight bound?
	%
	%
	\item What is the straight-line upward planar slope number of upward planar digraphs?
\end{inparaenum}

\smallskip
{\footnotesize \noindent \textbf{Acknowledgments.} Research partially supported by project: 
	``Algoritmi e 
	sistemi di analisi visuale di 
	reti complesse e di grandi dimensioni - Ricerca di Base 
	2018, Dipartimento di Ingegneria, Universit\`a degli Studi di  Perugia''.}

\bibliographystyle{splncs04}
\bibliography{paper}

\clearpage

\appendix

\section*{Appendix}

\section{Missing proofs of Section~\ref{sec:1bend}}

\begin{figure}
	\centering
	\begin{subfigure}{0.32\textwidth}
		\centering
		\includegraphics[width=\textwidth,page=1]{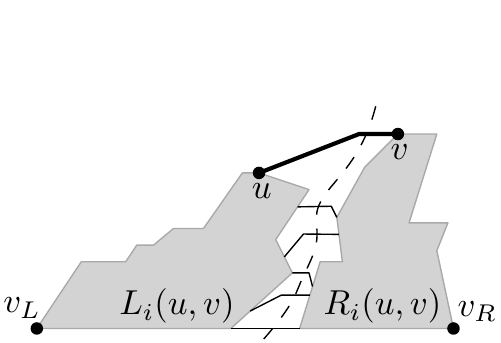}
		\caption{}
		\label{fi:cut1}
	\end{subfigure}
	\hfil
	\begin{subfigure}{0.32\textwidth}
		\centering
		\includegraphics[width=\textwidth,page=2]{figs/cut}
		\caption{}
		\label{fi:cut2}
	\end{subfigure}
	\hfil
	\begin{subfigure}{0.32\textwidth}
		\centering
		\includegraphics[width=\textwidth,page=3]{figs/cut}
		\caption{}
		\label{fi:cut3}
	\end{subfigure}
	\caption{Illustration for Lemma~\ref{le:stretching}.}
\end{figure}

\setcounter{lemma}{2}

\begin{lemma}
	Let $\auGamma_i$ be a drawing of $\auG^-_i$ that satisfies Invariants \myref{I1}--\myref{I5}. Let $(u,v)$ be any edge of $\auP_i$ such that $u$ is encountered before $v$  along $\auP_i$ when going from $v_L$ to $v_R$, and let $\lambda$ be a positive number. There exists a drawing $\auGamma'_i$ of $\auG^-_i$ that satisfies Invariants \myref{I1}--\myref{I5} and such that: (i) the horizontal distance between $u$ and $v$ is increased by $\lambda$; (ii) the horizontal distance between any two other consecutive vertices along $\auP_i$ is the same as in $\auGamma_i$.
\end{lemma}
\begin{proof}
	We prove by induction on $i$ that there exists a cut $(L_i(u,v),R_i(u,v))$ such that: (i) the vertices of the subpath of $\auP_{i}$ from $v_L$ to $u$ belong to $L_i(u,v)$, while the vertices of the subpath of $\auP_{i}$ from $v$ to $v_R$ belong to $R_i(u,v)$; (ii) every edge that crosses the cut has a horizontal segment (see Figure~\ref{fi:cut1}).  
	
	This is trivially true for each edge $(u,v)$ in $\auP_2$. Assume that it is true for $\auGamma_{i-1}$ by induction ($i > 2$). Let $u_1,u_2,\dots,u_q$ be the neighbors of $v_i$ in $\auP_{i-1}$; observe that $\auP_i$ is obtained from $\auP_{i-1}$ by replacing the subpath $\langle u_1, u_2,\dots, u_q \rangle$ with $\langle u_1, v_i, u_q \rangle$. Consider an edge $(u,v)$ of $\auP_{i}$. If $(u,v)$ also belongs to $\auP_{i-1}$ then $\langle u_1, u_2,\dots, u_q \rangle$ all belong to either $L_{i-1}(u,v)$ or to $R_{i-1}(u,v)$, say to $L_{i-1}(u,v)$ (see also Fig.~\ref{fi:cut2}). This implies that the cut $(L_{i}(u,v),R_i(u,v))$ with  $L_{i}(u,v)=L_{i-1}(u,v) \cup \{v_i\}$ and  $R_i(u,v)=R_{i-1}(u,v)$ satisfies (i) and (ii). If $(u,v)$ does not belong to $\auP_{i-1}$ then $(u,v)$ is either $(u_1,v_i)$ or $(v_i,u_q)$ (see Fig.~\ref{fi:cut3}). Suppose it is  $(u_1,v_i)$ (the other case is similar). Consider the cut $(L_{i-1}(u_1,u_2),R_{i-1}(u_1,u_2))$. Since by \myref{I2} $(u_1,v_i)$ contains an horizontal segment and there is a face of $\auGamma_{i}$ that contains both $(u_1,u_2)$ and $(u_1,v_i)$, the cut 
	$(L_{i}(u,v),R_i(u,v))$ with  $L_{i}(u,v)=L_{i-1}(u,v)$ and  $R_i(u,v)=R_{i-1}(u,v) \cup \{v_i\}$ satisfies (i) and (ii).
	
	We construct the drawing $\auGamma'_i$ from $\auGamma_i$ by increasing the length of the edges that cross the cut $(L_i(u,v),R_i(u,v))$ by $\lambda$ units. In other words, we increase by $\lambda$ units the $x$-coordinate of all vertices in $R_i(u,v)$. It is immediate to verify that $\auGamma'_{i}$ is a $1$-bend upward drawing whose real edges use only slopes in $\Ss$ and that $\auGamma'_{i}$ satisfies \myref{I2}--\myref{I4}. 
	
	Concerning planarity, we claim that translating the subdrawing induced by $R_i(u,v)$ does not violate planarity. Suppose for a contradiction that $\auGamma'_{i}$ is not planar. Then there exist two non-adjacent edges $e_L$ and $e_R$ of $\auP_i$ that intersect in a point $p$. Since $e_L$ and $e_R$ did not cross in $\auGamma_i$,  $p$ belongs to only one of the two edges in $\auGamma_i$, say $e_L$, and there exists a point $p'$ of $e_R$ in $\auGamma_{i}$ that has been translated to $p$ when transforming $\auGamma_i$ into $\auGamma'_i$. This means that $p'$ is encountered before $p$ when walking from left to right along the horizontal line $\ell$ passing through $p$ and $p'$. Since $p'$ has been translated and $p$ has not, $p'$ belongs to the subpath of $\auP_i$ that goes from $v$ to $v_R$, while $p$ belongs to the subpath of $\auP_i$ from $v_L$ to $v$. In other words, when walking along $\auP_i$ from $v_L$ to $v_R$, $p$ is encountered before $p'$. But this contradict Invariant \myref{I5} for $\auGamma_i$, which means that $\auGamma'_{i}$ is planar.
	
	We finally prove that \myref{I5} holds for $\auGamma'_{i}$. Let $\ell$ be any horizontal line and let $p$ and $p'$ be any two intersection points between $\ell$ and the polyline representing $\auP_i$ in $\auGamma_{i}$, with $p$ to the left of $p'$ along $\ell$. If, in $\auGamma'_{i}$, the order of $p$ and $p'$ along $\ell$ is reversed or $p$ and $p'$ coincide, then $p$ has been translated while $p'$ has been not. On the other hand, since Invariant \myref{I5} holds for $\auGamma_{i}$, $p$ precedes $p'$ when walking along $\auP_{i}$ from $v_L$ to $v_R$ and therefore if $p$ is translated, $p'$ is also translated. Thus, \myref{I5} holds for $\auGamma'_{i}$.
\end{proof}

\begin{lemma}
	Let $\auGamma_i$ be a drawing of $\auG^-_i$ that satisfies Invariants \myref{I1}--\myref{I5}. Let $u$ be a vertex of $\auP_i$, and let $t_u$ be any outer top ray of $u$ that crosses an edge of $\auG^-_i$ in $\auGamma_i$. There exists a drawing $\auGamma'_i$ of $\auG^-_i$ that satisfies Invariants \myref{I1}--\myref{I5} in which $t_u$ does not cross any edge of $\auG^-_i$.
\end{lemma}
\begin{proof}
	The ray $t_u$ can cross the subpath of $\auP_i$ from $v_L$ to $u$ and/or the subpath of $\auP_i$ from $u$ to $v_R$. Let $v$ and $w$ be the vertices that are encountered before and after $u$ along $\auP_{i}$ when going from $v_L$ to $v_R$, respectively. To remove the crossing(s) it is sufficient to apply Lemma~\ref{le:stretching} to $(v,u)$ and/or to $(u,w)$ for a sufficiently large $\lambda$; see Figs.~\ref{fi:stretch1}-\ref{fi:stretch2} for an illustration.    
\end{proof}

\setcounter{theorem}{4}
\begin{lemma}
	$\auGamma_2$ satisfies Invariants \myref{I1}--\myref{I5}.
\end{lemma}
\begin{proof}
	Invariants \myref{I1}, \myref{I2}, and \myref{I5} trivially hold. Since $\nr{v_L}{2}=\nr{v_R}{2}=0$, \myref{I3} trivially holds for $v_L$ and $v_R$. Vertices $v_1$ and $v_2$ are connected by a real edge and therefore each of them has at most $\Delta-1$ real incident edges that are not in $\auG_2$. Since there are $\Delta-1$ real top rays and they are all outer, \myref{I3} holds also for $v_1$ and $v_2$. \myref{I4} holds because all dummy top rays~are~outer. 
\end{proof}

\setcounter{lemma}{7}
\begin{lemma}
	Drawing $\Gamma$ can be computed in $O(n)$ time.
\end{lemma}
\begin{proof}
A bitonic $st$-ordering $\bo$ of $G$ can be computed in $O(n)$ time~\cite{DBLP:conf/gd/Gronemann16}, and the same time complexity suffices to compute a canonical augmentation $G'$ of $G$ by Lemma~\ref{le:gronemann}. A triangulated planar $st$-graph $\auG$ can be obtained in $O(n)$ time by augmenting $G'$ as follows.  For every non-triangular face $f$ of $G'$, let $t_f$ be the (unique) sink in $f$. For each vertex $u$ of $f$ different from $t_f$ we add the edge $(u,t_f)$ inside $f$, unless this edge already belongs to the boundary of $f$. Clearly, all these edges can be added in a planar way and each face of the resulting digraph is triangular. By construction, for any edge $(u,v)$  added to triangulate $G'$ there exists a directed path from $u$ to $v$ in $G'$. Thus $(u,v)$ does not create any directed cycle and $\auG$ has a single source (vertex $v_L$) and a single sink (vertex $v_n$) that are the same as in $G'$. Therefore, the triangulated digraph $\auG$ is a planar $st$-graph. The construction of $\auGamma_i$ from $\auGamma_{i-1}$, for $i=3,4,\dots,n$,  requires $O(deg(v_i))$ applications of Lemma~\ref{le:stretching}, where $deg(v_i)$ is the degree of $v_i$ in $\auG^-_i$. Since a straightforward implementation of the technique of Lemma~\ref{le:stretching} takes $O(n)$ time, and since $\sum_{i=3}^nO(deg(v_i))=O(n)$, the overall time complexity~would~be~$O(n^2)$.
	
	To achieve $O(n)$ time, the algorithm can be implemented to work in two phases. In the first phase, the exact coordinates of the vertices are not computed, but we only store information on the edges to reconstruct these coordinates in the second phase. More precisely, in the first phase, we consider each vertex $v_i$ according to $\co$ and assign to each edge $e=(u,v)$ that is in $\auG^-_i$ but not in $\auG^-_{i-1}$ two pairs of numbers $(s(e,u),l(e,u))$ and $(s(e,v),l(e,v))$. We aim at guaranteeing the following properties for each $i=2,3,\dots,n$:
	
	\begin{description}
		
		\item[P1] $s(e,u)$  (resp.\ $s(e,v)$) represents the slope of the segment of $e$ incident to $u$ (resp.\ $v$) in $\auGamma_i$.
		
		\item[P2] $l(e,u)$  (resp.\ $l(e,v)$) represents the length of the segment of $e$ incident to $u$ (resp.\ $v$) in $\auGamma_i$  if at least one of the following two conditions apply: (a) $e$ is on the boundary of $\auP_i$, or (b) $s(e,u) \neq 0$ (resp.\ $s(e,v) \neq 0$), i.e., the segment does not use the horizontal slope.
	
	\end{description}
	
	For each edge $e=(u,v)$ in $\auGamma_2$, we set $s(e,u)=s(e,v)=0$ and $l(e,u)=l(e,v)=1$. It is immediate to verify that Properties \myref{P1}--\myref{P2} hold. Let $v_i$, $i>2$, be the next vertex to be considered. Let $u_1,u_2,\dots,u_q$ be the neighbors of $v_i$ along $\auP_{i-1}$ and assume, by induction, that \myref{P1}--\myref{P2} hold for $i-1$. Consider any edge $e_j=(u_j,v_i)$, for $j=1,\dots,q$. We choose the slopes for the two segments of $e_j$ as explained above and set $s(e_j,u_j)$ and $s(e_j,v_i)$ accordingly (the choice of the slopes can be done in $O(1)$ time). In order to compute $l(e_j,u_j)$ and $l(e_j,v_i)$, it is sufficient to compute the positions of $u_1,u_2,\dots,u_q$ \emph{relative} to $u_1$. This can be done in $O(deg(v_i))$ time, because, by \myref{P1}--\myref{P2}, we know both the slope and the length of each edge segments along $\auP_{i-1}$ from $u_1$ to $u_q$. We can then calculate, again in $O(deg(v_i))$ the (relative coordinates of) the intersection points $p_1,p_2,\dots,p_q$. Afterwards, we may need to (repeatedly) apply Lemma~\ref{le:stretching}. Note that the application of this lemma does not modify the slope of any edge segment, and thus it preserves Property \myref{P1} for all the edges of $\auGamma_{i-1}$. Instead, it changes the length of some horizontal segments. However, all the involved horizontal segments do not belong to $\auP_{i}$. Finally, in order to set $l(e_j,u_j)$ and $l(e_j,v_i)$,  it suffices to know the values of $\lambda$ used in all the applications of Lemma~\ref{le:stretching}, which can be computed by only looking at the intersection points $p_1,p_2,\dots,p_q$. It follows that \myref{P1}--\myref{P2} hold. 
	
	Once $v_n$ has been considered, we have information on all the edges of $\auGamma_n=\auGamma$. In particular, by \myref{P1}--\myref{P2}, we know the slope and the length in $\auGamma$ of all the edges that do not contain any horizontal segment. These edges form a tree\footnote{In fact, this tree is one of the three trees obtained by a Schnyder decomposition~\cite{Schnyder1990}.} rooted at $v_n$ spanning the graph obtained removing $v_L$ and $v_R$ from $\auG$. To see this, observe that, for $j=1,2,\dots,n-1$, each vertex $v_j$  has been connected exactly once to a vertex $v_z$, with $j < z$, with an edge that does not contain any horizontal segment, as otherwise $v_j$ would belong to the outer face of $\auGamma$. Hence, $v_z$ is the (only) parent of $v_j$ in the spanning tree. Furthermore, $v_n$ is incident to at least one of these edges since it has degree at least three and exactly two of its edges contain a horizontal segment. Therefore, an assignment of valid coordinates to the vertices of $G$ can be obtained through a pre-order visit of this spanning tree (recall that for all the edges of the spanning tree we know the slope and length of its two segments). Finally, all edges that contain a horizontal segment (including those that are incident to $v_L$ and $v_R$) can be drawn as we know the slopes of both segments and the $y$-coordinate of the bend point.
\end{proof}

\begin{figure}
\centering
\begin{subfigure}{0.2\textwidth}
	\centering
	\includegraphics[width=\textwidth,page=1]{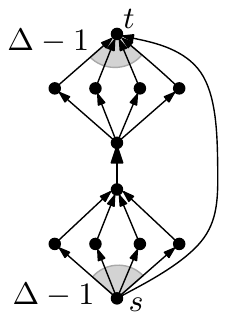}
	\caption{}
	\label{fi:lowerbound}
\end{subfigure}
\hfil
\begin{subfigure}{0.3\textwidth}
	\centering
	\includegraphics[width=\textwidth,page=1]{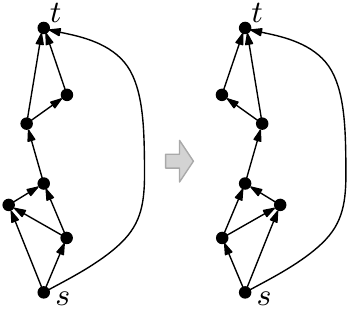}
	\caption{}
	\label{fi:cubic-bitonic}
\end{subfigure}
\caption{(a) A graph that requires $\Delta$ slopes and angular resolution at most $\frac{\pi}{\Delta}$ in every $1$-bend upward planar drawing. (b) Illustration for the proof of Theorem~\ref{thm:cubic-bitonic}.}
\end{figure}

\setcounter{corollary}{0}
\begin{corollary}
	Every bitonic $st$-graph with maximum vertex degree $\Delta \geq 2$ has 1-bend upward planar slope number at most $\Delta$, which is worst-case optimal.
\end{corollary}
\begin{proof}
	Every bitonic $st$-graph has $1$-bend upward planar slope number at most $\Delta$ (Theorem~\ref{thm:bitonic}). Also, for every $\Delta \geq 2$ there is a bitonic $st$-graph (shown in Fig.~\ref{fi:lowerbound}) that requires at least $\Delta$ slopes in any $1$-bend upward planar drawing~\cite{DBLP:conf/gd/GiacomoLM16}.
\end{proof}

\setcounter{theorem}{1}
\begin{theorem}
	Every planar $st$-graph with maximum vertex degree $3$ has $1$-bend upward planar slope number at most $3$.
\end{theorem}
\begin{proof}
	Let $G$ be a planar $st$-graph with maximum vertex degree $3$. By Theorem~\ref{thm:bitonic}, if $G$ has a bitonic $st$-ordering then the statement follows. If $G$ does not contain any forbidden configuration, then it is bitonic (see Section~\ref{sec:preliminaries}). Otherwise, recall that a forbidden configuration consists of at least three outgoing edges incident to the same vertex. Thus, $G$ contains exactly one forbidden configuration, which involves its source $s$ and the edge $(s,t)$. We can remove this forbidden configuration by mirroring the embedding of the subgraph obtained removing $(s,t)$ from $G$ (see Fig.~\ref{fi:cubic-bitonic} for an illustration).
\end{proof}

\section*{Missing proofs of Section~\ref{sec:2bend}}

\setcounter{lemma}{8}
\begin{lemma}
	Let $G=(V,E)$ be an $n$-vertex planar $st$-graph that is not bitonic. Let $G_b=(V_b,E_b)$ be an $N$-vertex bitonic subdivision of $G$, with a bitonic $st$-ordering $\bo=\{v_1,v_2,\dots,v_N\}$. There exists a planar $st$-graph $G'=(V',E')$ with an $st$-ordering $\co=\{v_L,v_R,v_1,v_2,\dots,v_N\}$ such that:
	
	\begin{inparaenum}[(i)]
		\item $V'=V_b\cup\{v_L,v_R\}$; 
	 	
	 	\item $E_b \subset E'$ and $(v_L,v_R) \in E'$; 
	 	
	 	\item $v_L$ and $v_R$ are on the boundary of the outer face of $G'$; 
	 	
	 	\item Every vertex of $G_b$ with less than two predecessors in $\bo$ has exactly two predecessors in $\co$.
	 	
	 	\item There is no vertex in $G'$ such that its leftmost or its rightmost incoming edge is an upper stub. 
	 	
	\end{inparaenum}
	
	\noindent Also, $G'$ and $\co$ are computed~in~$O(n)$~time.
\end{lemma}
\begin{proof}
We construct $G'$ together with its embedding by adding a vertex per step according to $\co$. We start with the $3$-cycle whose edges are $(v_L,v_R)$, $(v_L,v_1)$, and $(v_R,v_1)$ embedded so that, starting from the outer face, the edge $(v_R,v_1)$ is the first edge in the clockwise circular order of the edges around $v$. Let $G'_i$ be the plane digraph induced by $\{v_L,v_R,v_1,\dots,v_i\}$. The neighbors of $v_i$ that are in $G'_{i-1}$ all belong to the boundary of the outer face (because $\bo$ is an $st$-ordering of $G_b$). Thus $v_i$ can be planarly connected to its neighbors in $G'_{i-1}$. In order to guarantee properties (\ref{p:predecessors2}) and (\ref{p:upperstub}) we may add dummy edges connecting $v_i$ to some vertex of the outer face of $G'_{i-1}$ that is not adjacent to $v_i$ in $G_b$. 

Suppose first that $v_i$ has only one predecessor $u$ in $\co$. We claim that all the edges connecting $u$ to vertices that are after $v_i$ in $\co$ appear consecutively either in clockwise or in counterclockwise order around $u$ starting from $(u,v_i)$ in the embedding of $G_b$. If this was not the case, then there would be two vertices $v_j$ and $v_h$, with $i < j$ and $i < h$, such that $v_j$ precedes $v_i$ and $v_h$ follows $v_i$ in the circular order around $u$. But this would imply that $v_j,v_i,v_k$ form a forbidden configuration for $G_b$, which would contradict the fact that $G_b$ is bitonic. If all these edges appear after $v_i$ in clockwise (resp.\ counterclockwise) order, then we can add the edge $(v_i,w)$, where $w$ is the vertex preceding (resp.\ following) $u$ when walking clockwise along the boundary of $G'_{i-1}$ (see, for example, Fig.~\ref{fi:support1}).

\begin{figure}
	\centering
	\begin{subfigure}{0.4\textwidth}
		\centering
		\includegraphics[width=\textwidth,page=1]{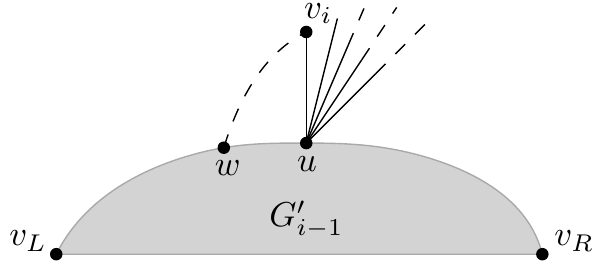}
		\caption{}
		\label{fi:support1}
	\end{subfigure}
	\hfil
	\begin{subfigure}{0.4\textwidth}
		\centering
		\includegraphics[width=\textwidth,page=2]{figs/support}
		\caption{}
		\label{fi:support2}
	\end{subfigure}
	\caption{\label{fi:support}Addition of dummy edges to guarantee properties~\ref{p:predecessors2} and~\ref{p:upperstub} of Lemma~\ref{le:gronemann2}. (a) $v_i$ has a single predecessor. (b) $v_i$ has a leftmost incoming edges that is an upper stub (the squared vertex is a dummy vertex).}
\end{figure}

Suppose now that $v_i$ has more than one predecessor in $\co$, but its leftmost (resp.\ rightmost) incoming edge $(u,v_i)$ is an upper stub. This means that $u$ is a dummy vertex and therefore it has no successor in $\bo'$ other than $v_i$. Then we can add the edge $(v_i,w)$, where $w$ is the vertex preceding (resp.\ following) $u$ when walking clockwise along the boundary of $G'_{i-1}$ (see, for example, Fig.~\ref{fi:support2}). Properties (i), (ii), and (iii) immediately follow from our construction, while (iv) and (v) are guaranteed by the dummy edges that we insert as explained above. Since $G'$ has $O(n)$ vertices and edges, the above procedure can be implemented to run in $O(n)$ time.
\end{proof}

\setcounter{theorem}{3}

\begin{theorem}
	Let $\Ss$ be any set of $\Delta+1$ slopes including the horizontal slope and let $G$ be an $n$-vertex upward planar graph with maximum vertex degree $\Delta \ge 2$. Graph $G$ has a $2$-bend upward planar drawing using only slopes in $\Ss$.
\end{theorem}
\begin{proof}
Since $G$ is upward planar, it is possible to augment $G$ with dummy edges to a planar $st$-graph $G_{st}$ with maximum vertex degree $\Delta_{st}$ (see, e.g.,~\cite{DBLP:journals/tcs/BattistaT88}). If we apply the algorithm of Theorem~\ref{thm:planarst} to $G_{st}$ we obtain a $2$-bend upward planar drawing using any set of $\Delta_{st}$ slopes including the horizontal one. However, it is not immediate to augment $G$ so that $\Delta_{st} \leq \Delta+1$.  On the other hand, the algorithm can be applied so to draw the dummy edges of $G_{st}$ using dummy slopes. In this case we should take into account the fact that a vertex $v$ that is a source (resp.\ a sink) in $G$ and not in $G_{st}$ may have $\Delta$ outgoing (resp.\ incoming) real edges. To cope with this issue it suffices to use any set of $\Delta+1$ real slopes that includes~the~horizontal~one.
\end{proof}

\end{document}